\documentclass[11pt,a4paper]{article}
\usepackage{amsthm}
\usepackage{amsmath}
\usepackage{amssymb}
\usepackage{times}
\usepackage{bm}
\usepackage{natbib}
\usepackage{color,tikz,subfigure}      
\usepackage[plain]{algorithm2e}
\usepackage{multirow}

\newcommand{\T}{T}

\newcommand{\fac}[1][k]{a^{(#1)}}
\newcommand{\facall}{a}
\newcommand{\datc}[1][k]{n^{(#1)}}
\newcommand{\Op}{\mathcal{O}}
\newcommand{\mav}{\bar{a}}
\newcommand{\ma}[1][k]{\mav^{(#1)}}
\newcommand{\wmy}{\tilde{y}}
\newcommand{\law}{\mathcal{L}}
\newcommand{\claw}[2]{\mathcal{L}\left\{#1 \mid  #2\right\}}
\newcommand{\fclaw}[1]{\mathcal{L}\left\{ #1 \mid  \cdot \, \right\}}

\newtheorem{theorem}{Theorem}
\newtheorem{proposition}{Proposition}

\newtheorem{lemma}{Lemma}

\begin{document}


\title{
Scalable inference for crossed random effects models
}

\author{O. Papaspiliopoulos\footnote{
Instituci\'o Catalana de Recerca i Estudis Avan\c{c}ats,  Ramon
  Trias Fargas 25-27, Barcelona
  08005. omiros.papaspiliopoulos@upf.edu }
, G.O. Roberts\footnote{
Department of Statistics, University of Warwick, Coventry, CV4 7AL, UK.
gareth.o.roberts@warwick.ac.uk}
\ and G. Zanella\footnote{
Department of Decision Sciences, BIDSA and IGIER, Bocconi University, via Roentgen 1, 20136
Milan, Italy. giacomo.zanella@unibocconi.it}
}

\maketitle

\begin{abstract}
  We analyze the complexity of Gibbs samplers for inference in crossed random effect models used in modern analysis of variance. We demonstrate that for certain designs the plain vanilla Gibbs sampler is not scalable, in the sense that its complexity is worse than proportional to the number of parameters and data. We thus propose a simple modification leading to a collapsed Gibbs sampler that is provably scalable. Although our theory requires some balancedness assumptions on the data designs, we demonstrate in simulated and real datasets that the rates it predicts match remarkably the correct rates in cases where the assumptions are violated. We also show that the collapsed Gibbs sampler, extended to sample further unknown hyperparameters, outperforms significantly alternative state of the art algorithms. 
\end{abstract}


\section{Introduction}

Crossed random effect models are additive models that relate a response variable to categorical predictors. In the literature they appear under various names, e.g. crossclassified data, variance component models or multiway analysis of variance. They provide the canonical framework for understanding the  relative importance of different sources of variation in a data set as argued in \cite{gelman-anova}. For the purposes of this article we focus on linear models according to which
\begin{align}
  \label{eq:anova}
y_{i_1\cdots i_K} \sim N\left\{a^{(0)}+a^{(1)}_{i_1}+\dots+a^{(K)}_{i_K}, (n_{i_1\cdots i_K} \tau_0)^{-1}\right\}, \quad i_k =1,\ldots,I_k, \quad k=1,\ldots,K
\end{align}
where  $\fac$ is a variance component, i.e., the vector of $I_k$ levels, $\fac_{i_k}$, for the $k$-th categorical factor; $\fac[0]$ is a global mean with $I_0=1$ level. These variance components might correspond to both main and interaction effects in categorical data analysis. We work with exchangeable Gaussian random effects, $\fac_j \sim N(0,1/\tau_k)$, for $k>0$,  which is by far the most standard choice, although interesting alternative priors exist as in  \cite{hoff-anova}.
We call a data design one with balanced levels if the same number of observations are made at each level of each factor, but this number can vary with factor, see Section \ref{sec:notation} for a mathematical definition.
The design has balanced cells if the same number of observations are available for each combination of factor levels, i.e., at each cell of the contingency table defined by the categorical predictors.
By construction, a design with balanced cells has also balanced levels.
In the notation of \eqref{eq:anova} we allow $n_{i_1\cdots i_K}=0$, which corresponds to empty cells.  The total number of factor levels, hence of regression parameters, is denoted by $p=\sum_{k=0}^K I_k$, and that of number of observations by $N= \sum_{i_1\cdots i_K}n_{i_1\cdots i_K}$.  Crossed random effect models adapt naturally to  modern high-dimensional but sparse data. For example, they are used in the context of recommender systems where in the simplest setup there are two factors, customers and products, and the response is a rating; the examples in  \cite{GaoOwen2017EJS} are such that $p \ll N \ll I_1 \times I_2$.

Likelihood-based inference for such models requires a marginalisation over the factors. An exact marginalisation is possible in the linear model due to the joint Gaussian distribution of responses and factors. However, this involves matrix operations the cost of which are $\Op(p^3)$, which is prohibitively large in modern applications. For example in the case of recommendation this cost this is typically $\Op(N^{3/2})$, hence infeasible for large datasets. The precision matrices of the Gaussian distributions involved can be computed efficienty, e.g., \cite{wilki04}, and may be sparse, hence black-box sparse linear algebra algorithms can be used for the matrix operations in the hope of reducing the complexity, but it has never been established that this is actually achieved in crossed random effect models. 

Alternatively, Markov chain Monte Carlo can be used to carry out the integration (and the inference more generally). The most popular and convenient algorithm in this context is the Gibbs sampler, which samples  the factors $\fac$ iteratively from their full conditional distributions.  Recently, \cite{GaoOwen2017EJS} sketched an argument that suggests that the complexity of this algorithm,  in the special case that $\fac[0]$ is assumed known, in the context of recommendation where $K=2$ with balanced cell design,  has complexity $\Op(N^{3/2})$. The complexity of a Markov chain Monte Carlo algorithm can be defined  as the product of the computational time per iteration and the number of iterations the algorithm needs to mix.
The heuristic argument in \cite{GaoOwen2017EJS} suggested that the Gibbs sampler is not scalable for crossed random effects models due to its superlinear cost in the number of observations. 

In this article we develop the theory for analysing the complexity of the Gibbs sampler for crossed random effect models under different designs. We propose a small modification of the basic algorithm, the collapsed Gibbs sampler, which we analyse too, and establish rigorously its superior performance and scalability. We obtain rigorous results on the mixing times of the algorithms in Section \ref{s:complexity-Gibbs} and analyse their computational cost in the Appendix. The careful analysis also shows that to an extent the scalability of the Gibbs sampler depends on the design. Our theory is useful beyond the specific designs that have been assumed to derive it.  The essence of the methodology we develop in this article is shown in Figure \ref{fig:two_fac_scaling}, details of which are given in Section \ref{sec:simstudy}.
\begin{figure}[h!]
\centering
\includegraphics[width=0.5\linewidth]{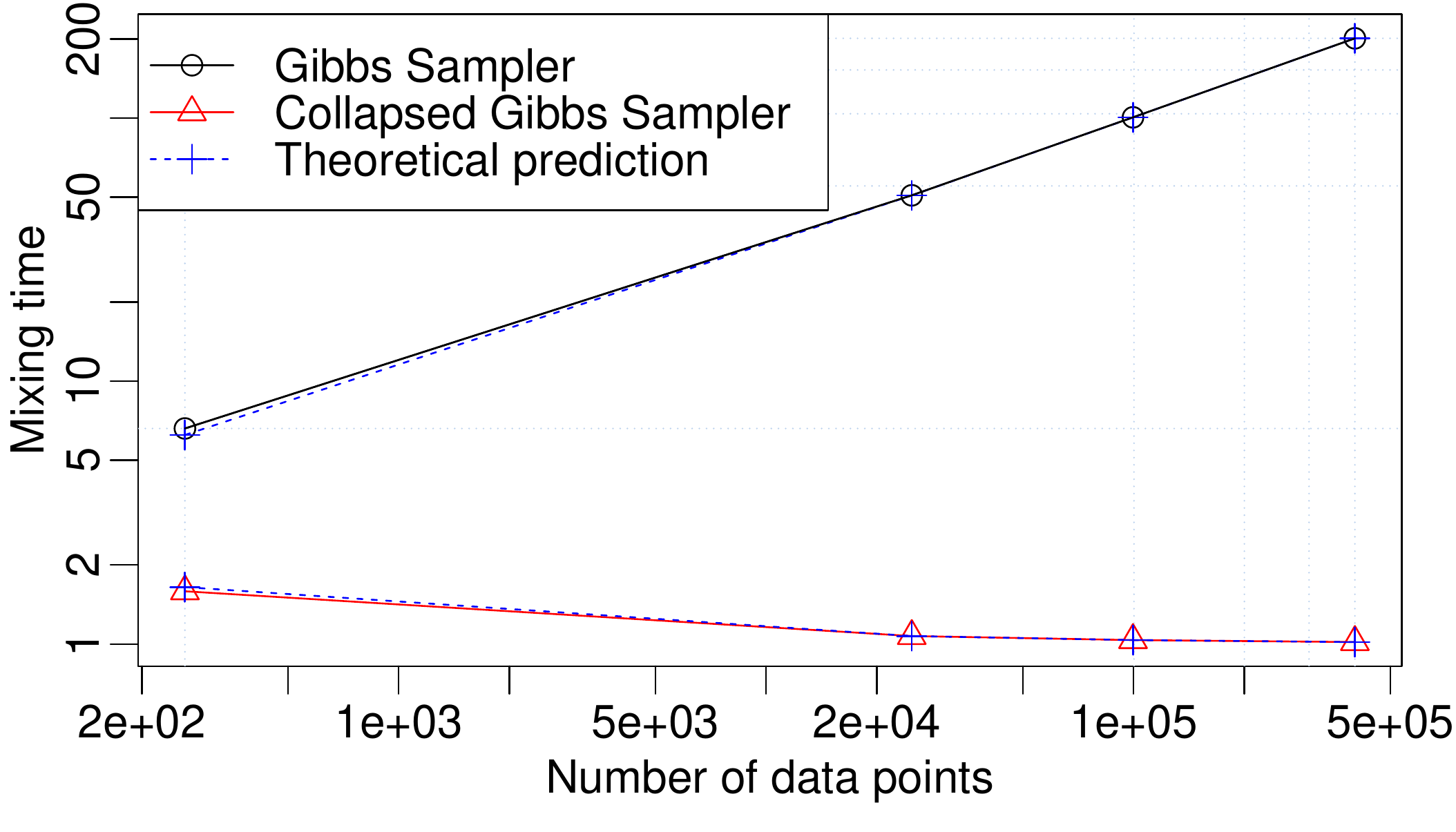}
\caption{Mixing time over number of datapoints for $K=2$, $I_1=I_2$ and $\tau_1=\tau_2=1$.
Data are missing at completely random with probability of missingness 0.9. The exact rates and those predicted by our theory (which, however, does not apply to these designs) are shown.}\label{fig:two_fac_scaling}
\end{figure}
The figure highlights different aspects of our results; we consider modern big data asymptotic regimes where both the number of parameters and observations grow, for $K=2$ factors; only a small fraction of the cells of the contingency table made by the outer product of the two factors is observed; the mixing time of the Gibbs sampler and the collapsed Gibbs sampler can be computed numerically and are plotted versus the size of the datasets, and it is evident the slowing down of the Gibbs sampler and the improvement of the collapsed Gibbs sampler with increasing data sizes; our theory is not applicable in these cases since the resultant designs, which have been generated randomly, are not balanced levels, still the rate that our theory predicts matches quite remarkably the correct rates. We obtain comparable results in a well-known real dataset of student evaluations with 5 factors in Section \ref{sec:InstEval}. For the student evaluations dataset we consider state of the art Markov chain Monte Carlo algorithms that sample the factor levels and the precision parameters in the crossed effect models, including parameter expansion and Hamiltonian Monte Carlo algorithms. We find that the collapsed Gibbs sampler we propose, appropriately extended to sample the precisions, has far superior performance. This is again a setting where our theory has not been developed yet (unbalanced designs, unknown precisions) but where the intuition gained from the simpler settings suggests practically useful algorithms. 
Therefore, our theory leads to generic guidelines for practitioners.

The theory is based upon a multigrid decomposition of the Markov chain generated by the sampler, which allows us to identify the slowest mixing components, and capitalises on existing theory for the convergence of Gaussian Markov chains. The multigrid decomposition  of a Markov chains is a powerful theoretical tool for studying its mixing time, since it provides its decomposition into independent processes. Identifying such decomposition is a kind of art; a previously successful example is in  \cite{zanella_mg} in the context of multilevel nested linear models. We point out that for nested hierarchical models scalable Bayesian computation can also be achieved with deterministic algorithms, such as belief propagation, as in  \cite{zanella_bp}.

The article closes with some conjectures that dictate our future research.

\section{Decompositions of the posterior distribution}

\subsection{Notation}\label{sec:notation}

The statistical model we work with is described in \eqref{eq:anova}. In accordance with standard practice Gaussian priors are used for the factor levels, $\fac_j \sim N(0,1/\tau_k)$, and an improper prior for the global mean, $p(\fac[0]) \propto 1$. When convenient we  write $\fac[0]_{i_0}$, which is is the same as $\fac[0]$. We allow $n_{i_1\cdots i_K}=0$, which corresponds to empty cells in the contingency table defined by the outer product of the categorical factors. With $n$ we denote the data incidence array, a multidimensional array with elements $n_{i_1\cdots i_K}$. Two-dimensional marginal tables extracted from the data incidence matrix are denoted by $\datc[l,k]$ and have elements $\datc[l,k]_{i_l,i_k}$, which is the total number of observations on level $i_l$ of factor $l$ and $i_k$ of factor $k$; margins of this table are denoted by $\datc$, and are vectors of size $I_k$ and elements $\datc_j$, which is the total number of observations with level $j$ on the $k$-th factor. By definition $\sum_j \datc_j=N$, where $N$ is the total number of observations. A data design has balanced levels if $\datc_j=N/I_k$ for every $k$ and $j$, and balanced cells if $n_{i_1\cdots i_K}=N/\prod_k I_k$ for all combinations of factor levels.

Averages of vectors are denoted by an overline, e.g., $\ma$; weighted averages are denoted by a tilde, e.g., $\wmy = \sum_{i_1\cdots i_K} y_{i_1\cdots i_K} n_{i_1\cdots i_K} /N$. The vector of all factor averages is denoted by $\mav$, the first element of which is trivially $\fac[0]$. Negative superscripts in factors denote the vector of factor levels for all factors but the one with the index whose negative value is used in the superscript, e.g., $\fac[-k]$ includes all factor levels except those of $\fac$. Similarly, negative subscripts in factors denote the vector of levels of the given factor except the level whose negative value is used in the subscript, e.g., $\fac_{-j}$ includes all levels of factor $k$ except the $j$-th level; $\facall$ denotes the vector of all levels of all factors. We define $\delta$ to be a residual operator that when applied to a vector returns the difference of its elements from their sample average, e.g., $\delta \fac$ has elements $\fac_j - \ma$ and is referred to as the factor's level increments; $\delta \facall$ denotes the vector of all such increments  (except $\delta \fac[0]$ which is 0 trivially).  

The law of a random variable $X$ is denoted by $\law(X)$, e.g., $\law (\fac_j) = N(0,1/\tau_k)$, and that of $X$ conditionally on $Y$ by $\law(X \mid  Y)$. 
When a joint distribution has been specified for $X$ and other random variables, $\fclaw{X}$ denotes the full conditional distribution of $X$ conditionally on the rest.  

\subsection{Full conditional distributions}
\label{sec:fclaw}

Fairly standard Bayesian linear model calculations yield that
\begin{align}
  \label{eq:fcondmu}
  \fclaw{\fac[0]} & = N\left \{\wmy - {\sum_k \sum_i \fac_i \datc_i \over N}, (N \tau_0)^{-1} \right \},
\end{align}
where the LHS is an example of the abridged notation we shall adopt for the full conditional distribution (in this case of $\fclaw{\fac[0]}$ given all other parameters and data).
With balanced levels this simplifies to
\begin{align}
  \label{eq:fcondmu-bl}
  \fclaw{\fac[0]} & = N\left \{\wmy - \sum_k \ma , (N \tau_0)^{-1} \right \}.
\end{align}
Similarly we obtain that for $k>0$
\begin{align}
  \label{eq:fconda}
  \fclaw{\fac_{j}} & =  N\left \{{\datc_j \tau_0 \over \datc_j \tau_0 + \tau_k} \left (\wmy^{(k)}_j - \fac[0] - {\sum_{l \neq k, l \neq 0} \sum_i \fac[l]_i \datc[k,l]_{j,i} \over \datc_j}\right), (\datc_j \tau_0 + \tau_k)^{-1} \right \},
\end{align}
where $\wmy^{(k)}_j$ is the weighted average of all observations for which their level on factor $k$ is $j$. With balanced cells this simplifies to 
\begin{align}
  \label{eq:fconda-cells}
  \fclaw{\fac_{j}} & =  N\left \{{N \tau_0 \over N \tau_0 + I_k \tau_k} \left (\bar{y} - \sum_{l \neq k} \ma[l] \right), I_k (N \tau_0 + I_k\tau_k)^{-1} \right \},
\end{align}

\subsection{Factorisations}

In balanced levels designs the posterior distribution of regression parameters admits certains factorisation, which are collected together in the following Proposition. 
\begin{proposition}\label{prop:b_l} For balanced levels designs
\begin{align*}
\claw{\mav,\delta \facall}{y}
& =
\claw{\mav}{y} \claw{\delta \facall}{y}\,,
\end{align*}
and
\begin{align}\label{eq:collapsed_iid}
  \claw{\ma[-0]}{y}  
=
 \prod_{k=1}^K \claw{\ma[k]}{y}  \,.
\end{align}
For balanced cells designs we have further
\begin{align*}
\claw{\delta \facall}{y}
& = \prod_{k} \claw{\delta \fac}{y}.
\end{align*}
\end{proposition}
The factorisation in \eqref{eq:collapsed_iid} is particularly relevant to the collapsed Gibbs sampler we introduce later in the article. 
A sketch of the proof of Proposition \ref{prop:b_l} is the following. For the first factorisation, directly from \eqref{eq:fconda} with the assumption of balanced levels we obtain that
\begin{align}
  \label{eq:fcondma}
\claw{\ma}{y,\fac[-k],\delta \fac} = N \left\{{N \tau_0 \over N \tau_0 + I_k \tau_k} \left(\wmy - \fac[0] - 
\sum_{l\neq k} \ma[l] \right) , (N \tau_0+I_k \tau_k)^{-1} \right \}.
\end{align}
We use the fact that global and local Markovian properties are equivalent, see, e.g., Section 3 of \cite{besag}. This yields the independence stated in the lemma. 
The proof of \eqref{eq:collapsed_iid} follows by similar arguments using $\claw{\ma[-0,-k]}{y}=N(0,(I_k\tau_k)^{-1})$. The third factorisation is argued in the same way  noting that  \eqref{eq:fconda-cells} implies that
\begin{align*}
\claw{\delta \fac_j}{y,\fac[-k], \delta \fac[-k]}
& = N\left \{0, (I_k-1) (N \tau_0 + I_k\tau_k)^{-1} \right \}.
\end{align*}

\section{Gibbs samplers for inference}
\label{sec:GS}

We consider two main algorithms in this paper. The first is a block Gibbs sampler that updates  in a single block the levels of a given factor conditionally on everything else. Due to the dependence structure in the model, the levels of a given factor conditionally on the rest are independent, hence in practice the sampling is done separately for each factor level, i.e., iteratively from $\fclaw{\fac_{i_k}}$, for $i_k =1,\ldots,I_k$, and $k=0,\ldots,K$; these distributions are specified in Section \ref{sec:fclaw}. We refer to this algorithm as the Gibbs sampler, although it should be understood that it is just one implementation of the scheme. 

We also consider the collapsed version of this algorithm that samples from $\claw{\fac[-0]}{y}$, i.e., the algorithm that is obtained by first analytically integrating out  the global mean $\fac[0]$, and then sampling in blocks the levels of each of the remaining factors; we term this algorithm the collapsed Gibbs sampler. In practice, we implement this algorithm by sampling iteratively from $\fclaw{\fac[0],\fac[k]}$, for $k=1,\ldots,K$. In this implementation we first sample $\claw{\fac[0]}{y,\fac[-0,-k]}$, and then $\fclaw{\fac_{i_k}}$ for $i_k =1,\ldots,I_k$ as in the previous scheme. The implementation of the collapsed Gibbs sampler relies on the  following result.
\begin{proposition}\label{prop:marginal_law}
Denoting $s^{(k)}_j=\datc[k]_j\tau_0/(\tau_k+\datc[k]_j\tau_0)$, then
\begin{align}
  \label{eq:fcondmu_collapsed}
  \claw{\fac[0]}{y,\fac[-0,-k]} & = N\left \{ 
  \frac{1}{\sum_j s^{(k)}_j}
  \sum_j s^{(k)}_j
  \left(\wmy^{(k)}_j-\frac{\sum_{l\neq k}\sum_{i} \fac[l]_i \datc[k,l]_{j,i}}{\datc[k]_j} \right)
  ,
  \frac{1}{\tau_k\sum_j s^{(k)}_j} 
\right \}.
\end{align}
\end{proposition}
%

The reason why we prefer to present the collapsed Gibbs sampler in this way where $\fac[0]$ is updated together with each block as opposed to being integrated out before sampling starts, is because our preferred version is still realisable in more elaborate models, e.g., generalised linear crossed random effects models. In such extensions exact sampling from $\fclaw{\fac[0],\fac}$ might not be feasible, but a Metropolis-Hastings step can be used instead. Additionally, it requires a minimal modification of the Gibbs sampler code to implement.

\section{Multigrid decomposition of the Gibbs samplers}

\subsection{Notation}

For the stochastic processes generated by Markov chain Monte Carlo the time index corresponds to iteration, which is generically denoted by $t$, and it is included in parentheses, e.g., $x(t)$; in such a case the stochastic process over $T$ iterations is denoted by $\{x(t)\}_{t=1}^T$; we write $\{x(t)\}$ when $T=\infty$;  we write $\{(x,z)(t)\}$ to denote  a stochastic process that at each time $t$ takes as value the vector composed by $x(t)$ and $z(t)$. We say that the stochastic process $\{x(t)\}$ is a timewise transformation of another $\{y(t)\}$ if there is a function $\phi$ such that $x(t) = \phi\{y(t)\}$ for all $t$.

\subsection{Main results}

The results we derive in this paper stem from the following result, the proof of which is given in the Appendix. 
\begin{theorem}
  \label{th:multigrid1} (Multigrid decomposition) Let $\{\facall(t)\}$ be the Markov chain generated either by the Gibbs sampler or the collapsed Gibbs sampler for balanced levels designs. Then, the timewise transformations  $\{\mav(t)\}$ and $\{\delta \facall(t)\}$ obtained from $\{\facall(t)\}$  are each a Markov chain and they are independent of each other.
\end{theorem}
A crucial point here is that the fact that the posterior distribution of $\mav$ and $\delta \facall$ factorise for balanced levels, as shown in Proposition \ref{prop:b_l}, does not  imply that the corresponding chains $\{\mav(t)\}$ and $\{\delta \facall(t)\}$ are independent of each other. The following very simple example that makes this point clear. Consider a Gibbs sampler that targets a bivariate Gaussian for $(x,y)$ with correlation $\rho$ and standard Gaussian marginals.  
Then the transformation $x$ and $z=y-\rho x$ orthogonalises the target, but the  corresponding stochastic processes $\{x(t)\}$ and $\{z(t)\}$ obtained by timewise transformation of the original chain $\{(x,y)(t)\}$ are not independent Markov chains, see, e.g., the cross-correlogram in  Figure \ref{fig:cross_corr}.
\begin{figure}[h!]
\centering
\includegraphics[width=0.5\linewidth]{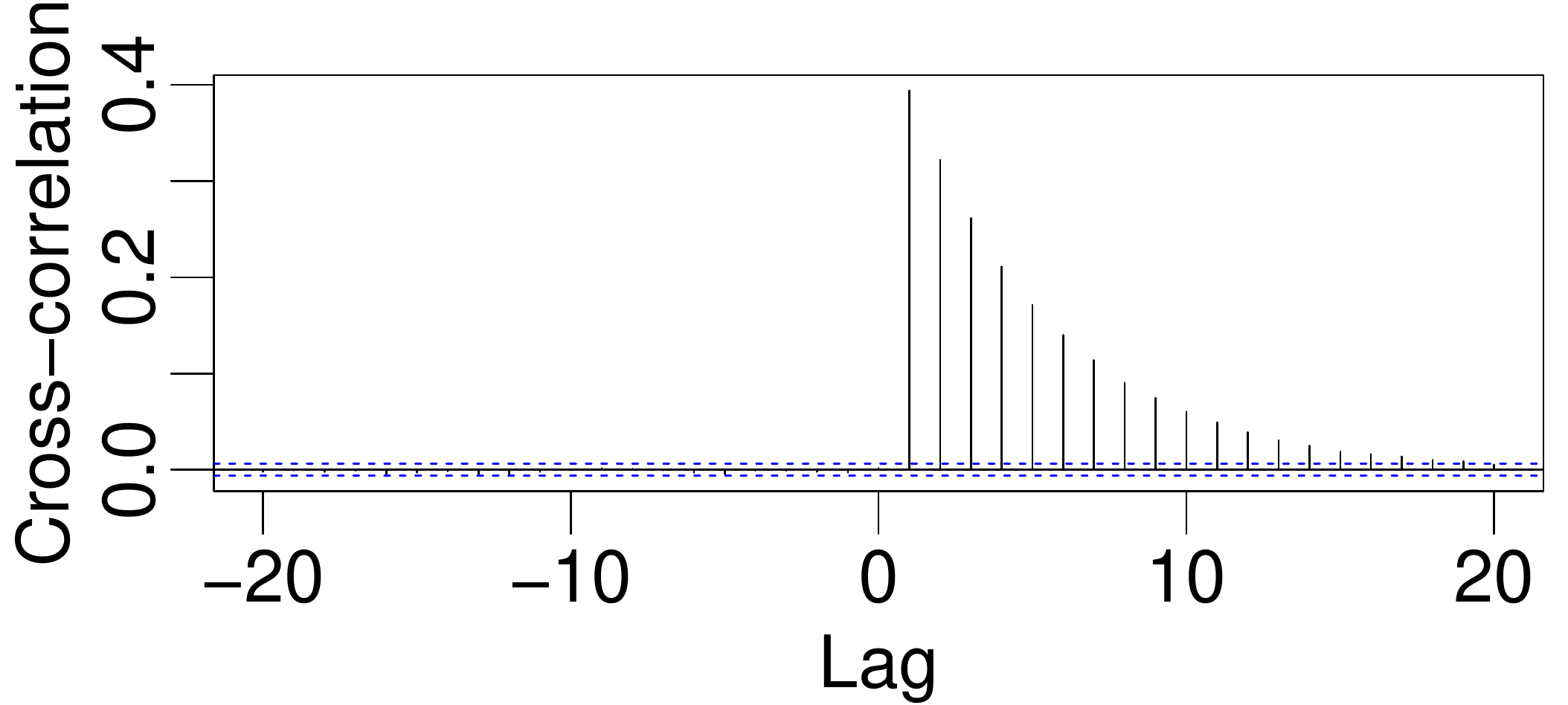}
\caption{Cross correlation between $\{x(t)\}$ and $\{z(t)\}$, where $z(t)=y(t)-\rho x(t)$, and $\{(x,y)(t)\}$ is the Gibbs sampler Markov chain on a bivariate Gaussian with correlation $\rho=0.9$.}\label{fig:cross_corr}
\end{figure}
Although this is a toy example, there are many instances where an independence factorisation of the target distribution does not imply that of the MCMC algorithm adopted (for example in  Hamiltonian Monte Carlo, population Markov chain Monte Carlo and piecewise deterministic Monte Carlo algorithms such as the zig-zag and bouncy particle sampler). 
There are subtle and deep reasons why the factorisation in Proposition \ref{prop:b_l} extends to the independence of the Markov chains obtained as timewise transformations.

In Section \ref{s:complexity-Gibbs} we use Theorem \ref{th:multigrid1} in conjunction with two others to characterise the complexity of the two samplers. The first of the additional results is about convergence rates of Markov chains, and shows how to relate the rate of convergence of $\{\facall(t)\}$ to that of the timewise transformations $\{\mav(t)\}$ and $\{\delta \facall(t)\}$. The second is about the rate of convergence of each of the Markov chains $\{\mav(t)\}$ and $\{\delta \facall(t)\}$.    

\section{Complexity analysis}
\label{s:complexity-Gibbs}

\subsection{Complexity of Markov chain Monte Carlo}
\label{complexity}

In this article we focus on $L^2(\pi)$ convergence,
which relies on functional analytic concepts, a very high level description of which are given below. For a given target distribution $\pi$ defined on a state space $\mathcal{X}$, we define $L^2(\pi)$ to be the space of complex-valued functions that are square-integrable with respect to $\pi$. We define the inner product in this space such that the associated norm of a function $f: \mathcal{X} \to R$ is $\|f\|^2=\int_{\mathcal{X}} f(x)^2 \pi(dx)$. For a Markov chain  $\{x(t)\}$ defined on $\mathcal{X}$ with transition kernel $P$ that is invariant with respect to $\pi$, we view $P$ as an integral operator on $L^2(\pi)$, and we say that it converges geometrically fast to $\pi$ in $L^2(\pi)$ norm (also known as operator norm), if and only if the second largest in absolute value eigenvalue of $P$, known as its geometric rate of convergence, is less than 1. The spectral gap of $P$ 
is defined as the difference between 1 and the rate of convergence, hence a Markov chain converges in $L^2(\pi)$ norm if and only it has positive spectral gap.  All this is fairly standard functional analysis theory applied to Markov chains on general state spaces. 

For our purposes, we define the mixing time of a Markov chain to be the inverse of its spectral gap; this can be interpreted as the number of iterations needed to subsample the Markov chain so that the resultant draws are roughly independent of each other. The complexity of a Markov chain Monte Carlo algorithm can be defined as the product of the mixing time and the cost per iteration.


\subsection{Timewise transformations and convergence of Markov chains}

The multigrid decomposition in Theorem \ref{th:multigrid1} identifies two timewise transformations of the Markov chain $\{\facall(t)\}$ produced by either of the algorithms considered in this article, each of which evolves independently of each other as a Markov chain. We can relate the rate of convergence of the Markov chains involved in this decomposition using the following two technical lemmata that are proved in the Appendix.

\begin{lemma}\label{lemma:transformation_MC}
Let $\{x(t)\}$ be a Markov chain with invariant distribution $\pi$ and $\{y(t)\}$ be a timewise transformation given by $y(t)=\phi(x(t))$, where $\phi$ is an injective function.
Then $\{y(t)\}$ is a Markov chain with the same rate of convergence as $\{x(t)\}$.
\end{lemma}

\begin{lemma}\label{lemma:independent_MC}
Let $\{x(t)\}$ be a Markov chain with state space $\mathcal{X}_1\times \mathcal{X}_2$ and target distribution $\pi_1\otimes\pi_2$. 
If the stochastic processes $\{x_1(t)\}$ and $\{x_2(t)\}$ obtained by projection on the $\mathcal{X}_1$ and $\mathcal{X}_2$ components are two independent Markov chains, then the rate of convergence of $\{x(t)\}$ equals the supremum between the rates of convergence of $\{x_1(t)\}$ and $\{x_2(t)\}$.
\end{lemma}

Therefore, for balanced levels designs the rate of convergence of the Markov chain $\{\facall(t)\}$, generated either by the Gibbs sampler or the collapsed Gibbs sampler, is the larger of the rates of the two chains $\{\mav(t)\}$ and $\{\delta \facall(t)\}$. Each of these chains is amenable to analysis using the theory summarised in Section \ref{sec:spectral-Gauss} below. 

\subsection{The spectral gap of the Gibbs sampler on Gaussian
  distributions}
\label{sec:spectral-Gauss}
The Markov chain $\{x(t)\}$ generated by a Gibbs Sampler targeting a Gaussian multivariate distribution $N(\mu,\Sigma)$ is a Gaussian autoregressive process evolving as $x(t+1)\mid x(t)\sim N(Bx(t)+b,\Sigma-B\Sigma B^\T )$, see for example Lemma 1 in \cite{roberts1997updating}. The details of the Gibbs sampler (e.g., the order that its components are updated or blocked together) are reflected in the precise form of $B$.
There is a generic recipe how to obtain $B$ described in Lemma 1 in \cite{roberts1997updating}, but sometimes it is easier to work it out directly from first principles, as for example we do in Propositions \ref{prop:averages_rate} and \ref{prop:bl_two_fac} and  below. This representation implies that the rate of convergence of the  Gibbs sampler is $\rho(B)$, the largest absolute eigenvalue of the matrix $B$, see Theorem 1 of \cite{roberts1997updating}. This characterisation of the $L^2(\pi)$ rate of convergence is immensely useful and has provided invaluable insights into the performance of the Gibbs sampler and has lead to much more efficient modifications of the basic algorithm, see for example \cite{papa2003,papaspiliopoulos2007general}. However,  in high-dimensional scenarios it is often very challenging to compute $\rho(B)$ explicitly as a function of the important parameters of the model (e.g., $p$ and $N$ in the crossed effects models considered here). Hence as a tool for understanding the complexity of the Gibbs sampler in difficult problems this approach has limited scope. In this article we will make it useful by combining it with the  multigrid decomposition of Theorem \ref{th:multigrid1}, which 
collapses the problem to studying the spectral gaps of the Gaussian subchains $\{\mav(t)\}$ and $\{\delta \facall(t)\}$ that turn out to be amenable to direct analysis.

\subsection{Complexity analysis for balanced cells designs}
\label{sec:balanced-cells}


The most substantial result of this section is Proposition \ref{prop:averages_rate} below, which actually holds for  balanced levels designs too, hence used also in Section \ref{sec:balanced-levels}. It characterises the rate of convergence of one of the two timewise transformations involved in the multigrid decomposition.

\begin{proposition}\label{prop:averages_rate}
For balanced levels designs the rate of convergence of the Markov chain $\{\mav(t)\}$ defined in Theorem \ref{th:multigrid1} equals $\max_{k}\frac{N\tau_{0}}{N\tau_{0}+I_k\tau_{k}}$ for the Gibbs Sampler  and $0$ for the collapsed Gibbs Sampler, and this rate is the same for any order that the different blocks are updated. 
\end{proposition}
\begin{proof}
For the Gibbs Sampler, the subchain $\{\mav(t)\}$ is a Gaussian Gibbs Sampler, with $(K+1)$ one-dimensional components. We can explicitly work out that its autoregressive matrix $B$ takes the form
  \begin{equation}
B=
\begin{pmatrix}\label{eq:B_matrix_form}
0 & -1 & \dots & -1\\
0 &  &  & \\
\vdots &  & T &\\
0  &  &  & 
\end{pmatrix}
\end{equation}
where $T$ is a $K\times K$ lower triangular matrix with diagonal elements equal to $(r_1,\dots,r_K)$, with
\begin{align}\label{eq:ratios}
r_k
=
\frac{N\tau_{0}}{N\tau_{0}+I_k\tau_{k}}\,.
\end{align} 
We check \eqref{eq:B_matrix_form} verifying directly that $\mathbb{E}[\mav(t+1)\mid \mav(t)]=B\mav(t)+b$. 
From equation \eqref{eq:fcondmu-bl} we have $\mathbb{E}[\fac[0](t+1)\mid \mav(t)]=\wmy - \sum_k \ma(t)$, which implies that the first row of $B$ is as in \eqref{eq:B_matrix_form}.
To conclude the proof of \eqref{eq:B_matrix_form} we need to show that
\begin{align}\label{eq:induction}
\mathbb{E}\left[\ma[k](t+1)\mid \mav(t)\right]=r_k\ma[k]+\sum_{l=1}^{k-1}T_{kl}\ma[l]+b_k
\end{align}
for some $(T_{kl})_{l<k}$ and $(b_k)_k$.
Using \eqref{eq:fcondma}, we have
\begin{align*}
\mathbb{E}&\left[\bar{a}^{(k)}(t+1)\mid \mav(t)\right]
\\=&
\mathbb{E}\left[\mathbb{E}\left[\bar{a}^{(k)}(t+1)\mid \fac[0](t+1),\ma[1](t+1),\dots,\ma[k-1](t+1),\ma[k+1](t),\dots,\ma[K](t)\right]\mid \mav(t)\right]=
\\=
&r_k\left(\wmy-\mathbb{E}\left[\fac[0](t+1)\mid \mav(t)\right]
- \sum_{l=1}^{k-1}\mathbb{E}\left[\bar{a}^{(l)}(t+1)\mid \mav(t)\right]
- \sum_{l=k+1}^K\bar{a}^{(l)}(t)
\right)
\\=
&r_k\left(\sum_{s=1}^k\ma[s](t)
- \sum_{l=1}^{k-1}\mathbb{E}\left[\ma[l](t+1)\mid \mav(t)\right]
\right)
\,.
\end{align*}
When $k=1$ the latter implies $\mathbb{E}[\ma[1](t+1)\mid \mav(t)]=r_1\ma[1]$, meaning that \eqref{eq:induction} holds for $k=1$.
By induction we have that \eqref{eq:induction} holds for all $k=1,\dots,K$.
In fact if \eqref{eq:induction} holds for $1$ up to $k-1$, we have 
$$
\mathbb{E}\left[\bar{a}^{(k)}(t+1)\mid \mav(t)\right]
=
r_k\left(\sum_{s=1}^k\ma[s](t)
- \sum_{l=1}^{k-1}r_l\ma[l]+\sum_{s=1}^{l-1}T_{ls}\ma[s]+b_l
\right)\,,
$$
meaning that \eqref{eq:induction} holds also for $k$. Therefore $B$ has a form as in \eqref{eq:B_matrix_form}.

Since $T$ is a lower triangular matrix its spectrum coincides with its diagonal elements $(r_1,\dots,r_K)$.
For each $k=1,\dots,K$, let $v^{(k)}$ be the eigenvector with eigenvalue $r_k$.
It is easy to check that the $(K+1)$-dimensional vector $w^{(k)}=(-r_k^{-1}\sum_{\ell=1}^K v^{(k)}_\ell,v^{(k)}_1,\dots,v^{(k)}_K)$ is an eigenvector of B with eigenvalue $r_k$.
Thus $(r_1,\dots,r_k)$ are also eigenvalues of $B$.
Finally note that $(1,0,\dots,0)$ is an eigenvector of $B$ with eigenvalue 0. With these ingredients the proof of the claim for the Gibbs sampler follows immediately. 

For the collapsed Gibbs Sampler, $\mav(t)$ is obtained from $\mav(t-1)$ by simulating  $\ma[k](t)$ from 
  \begin{align*}
    & \claw{\ma(t)}{y,\ma[1](t),\ldots,\ma[k-1](t),\ma[k+1](t-1),\ldots,\ma[K](t-1)}\,,
  \end{align*}
  for $k=1,\dots,K$.
  By Proposition \ref{prop:b_l}, the latter procedure produces independent and identically distributed draws from $\claw{\ma[-0]}{y}$, or equivalently $\claw{\mav}{y}$ if $\ma[0]$ is jointly updated with $\ma[k]$.

These rates do not depend on the order that the different components are updated. This is trivially true for the collapsed Gibbs since the components are independent. For the Gibbs sampler the argument is as follows. 
The Gibbs Sampler rate of convergence is invariant with respect to cyclic permutations of the order of update of the components, see e.g.\ \citet[p.297]{roberts1997updating}.
Thus we can always assume $\fac[0]$ to be the first component to be updated.
Then the result follows by relabeling the components $\fac[1]$ to $\fac[K]$ according to their update order and replicating the argument developed in the previous paragraphs. 
\end{proof}

The main result of this section follows rather easily from Proposition \ref{prop:averages_rate}. 
\begin{theorem}\label{thm:rate_gibbs_bc}
For balanced cells designs, the mixing time of the Gibbs Sampler is $1+\max_{k=1,\dots,K}\frac{N\tau_0}{I_k\tau_k}$, and that of the collapsed Gibbs Sampler is 1, 
i.e., it produces independent and identically distributed draws from the target, and these rates do not depend on the order that different components are updated.
\end{theorem}

\begin{proof}[of Theorem \ref{thm:rate_gibbs_bc}]
Let $\{\facall(t)\}$ be the Markov chain generated by the Gibbs Sampler or its collapsed version.
Lemma \ref{lemma:transformation_MC} implies that $\{(\mav,\delta \facall)(t)\}$ is a Markov chain with the same rate of convergence as $\{\facall(t)\}$.
Thus, by means of Theorem \ref{th:multigrid1} and Lemma \ref{lemma:independent_MC}, the rate of convergence of  $\{\facall(t)\}$ equals the maximum between the rate of convergence of $\{\mav(t)\}$ and the one of  $\{\delta\facall(t)\}$.
Proposition \ref{prop:b_l} implies that $\{\delta\facall(t)\}$ performs independent sampling from $\claw{\delta \facall}{y}$ and thus its rate of convergence is 0 and the rate of convergence of  $\{\facall(t)\}$ equals the one of  $\{\mav(t)\}$. 
To conclude, Proposition \ref{prop:averages_rate} and the definition of mixing times as inverse of the spectral gap imply the statement to be proved. 
\end{proof}

The theorem completely characterises the mixing time of the Gibbs sampler and the collapsed Gibbs sampler  for balanced cells designs. Considering the computational cost of the algorithms, we find that each of the algorithms requires an $\mathcal{O}(N)$ computation at initialisation to precompute data averages. In the Appendix we show that both algorithms have the same cost per iteration, which is proportional to the number of parameters, $p$. Therefore, the collapsed Gibbs sampler is  an $\mathcal{O}(p)$ implementation of exact sampling from the posterior.

We now consider asymptotic regimes.  The more classical asymptotic regime, which we will refer to as infill asymptotics, keeps the number of factors and levels fixed, hence 
$K$ and $p$ fixed, and increases the number of observations per cell, hence $N$ grows. The other more modern asymptotic regime, which we will refer to as outfill asymptotics, increases $p$ with $N$, e.g. considering the observations per cell bounded and increasing the number or levels and/or factors. It is this type of asymptotic that it is more interesting in recommendation applications.

Regardless of the asymptotic regime considered the mixing time of the collapsed Gibbs sampler is $\mathcal{O}(1)$. On the other hand, that of the Gibbs sampler depends on the regime considered. In infill asymptotics Theorem \ref{thm:rate_gibbs_bc} implies that the mixing time of the algorithm is $\mathcal{O}(N)$. An intuition for this deterioration of the algorithm with increasing data size can be obtained by considering the analysis of non-centered parameterisations for hierarchical models in  Section 2 of \citet{papaspiliopoulos2007general}; the parameterisation of the crossed effect model is non-centred and the infill asymptotics regime makes the data increasingly informative per random effect, hence we should anticipate the deterioration. Therefore, in this regime the complexity of both algorithms is $\mathcal{O}(N)$ but in practice the collapsed will be much more efficient. In outfill asymptotics, both $N$ and the number of factor levels $I_k$'s are growing, hence by Theorem \ref{thm:rate_gibbs_bc} the mixing time of the Gibbs sampler is no worse than $\mathcal{O}(N)$ but no better than  $\mathcal{O}(N^{1-1/K})$. The lower bound on the mixing time  can be deduced from the balanced cells design assumption, which implies $\prod_{k=1}^KI_k\leq N$ and $\min_{k}I_k\leq N^{1/K}$; the bound is achievable when $I_1=\dots=I_K$. On the other hand, the number of parameters can grow as different powers of $N$. For example, if the number of levels for all but one factor are fixed and those of the remaining factor are increasing (e.g., fixed number of customers, increasing number of products) then $p$ is $\mathcal{O}(N)$ and the mixing time of the Gibbs sampler is also $\mathcal{O}(N)$, resulting in a Gibbs Sampler complexity of $\mathcal{O}(N^{2})$, whereas the collapsed Gibbs sampler is $\mathcal{O}(N)$.



\subsection{Complexity analysis for balanced levels designs}
\label{sec:balanced-levels}

The strategy for obtaining complexity results for balanced levels designs is the same as for balanced cells and Proposition \ref{prop:averages_rate} is as instrumental. However, in this case the analysis is much more complicated since the second timewise transformation, $\{\delta \facall(t)\}$, does not sample anymore independently from its invariant distribution; in fact its invariant distribution does not factorise as in the case of balanced cells. On the other hand, Lemma \ref{lemma:independent_MC} and Proposition \ref{prop:averages_rate} imply  immediately lower bound on the mixing time of the Gibbs sampler.
\begin{theorem}\label{thm:rate_gibbs_bl}
For balanced levels designs, the mixing time of the Gibbs Sampler is at least $1+\max_{k=1,\dots,K}\frac{N\tau_0}{I_k\tau_k}$. 
\end{theorem}

From Proposition \ref{prop:averages_rate} we also know that the rate of the blocked Gibbs sampler is that of $\{\delta \facall(t)\}$. Therefore, obtaining explicit rates of convergence for  $\{\delta \facall(t)\}$ is the step needed for characterising the mixing time of both algorithms in balanced levels designs. We are able to do this for $K=2$ in Proposition \ref{prop:bl_two_fac} below. Our theory is based on an  auxiliary process $\{i(t)\}$ with discrete state space $\{1,\dots,I_1\}\times \{1,\dots,I_2\}$ that evolves according to a two component Gibbs Sampler, iteratively updating $i_1\mid i_2$ and $i_2\mid i_1$, with invariant distribution $p(i_1,i_2)= n_{i_1i_2}/N$. 

\begin{proposition}\label{prop:bl_two_fac}
For balanced levels designs with $K=2$, the rate of convergence of the Markov chain $\{\delta \facall(t)\}$ is
\begin{align*}
\frac{N\tau_{0}}{N\tau_{0}+I_1\tau_{1}}
\frac{N\tau_{0}}{N\tau_{0}+I_2\tau_{2}}
\rho_{aux}\,,
\end{align*}
where $\rho_{aux}$ is the rate of convergence of the auxiliary Gibbs sampler $\{i(t)\}$.
\end{proposition}

\begin{proof}
The chain  $\{\delta\facall(t)\}$ is a two-component Gibbs Sampler that alternates updates from the conditional distributions $\claw{\delta\fac[1]}{y,\delta\fac[2]}$ and $\claw{\delta\fac[2]}{y,\delta\fac[1]}$. Thus, $\{\delta\fac[1](t)\}$ is marginally a Markov chain and its rate of convergence equals the one of $\{\delta\facall(t)\}$, see e.g. \cite{roberts2001deinitializing}.
Let $B_1$ and $B_2$  defined by $\mathbb{E}[\delta\fac[1]\mid \delta\fac[2],y]=B_1 \delta\fac[2]+b_1$ and $\mathbb{E}[\delta\fac[2]\mid \delta\fac[1],y]=B_2\delta\fac[1]+b_2$. It is then a simple computation that 
$\delta\fac[1](t)$ is a Gaussian autoregressive process with autoregression matrix
$B_1B_2$.
Since for balanced levels design it holds ${\datc_j \tau_0 \over \datc_j \tau_0 + \tau_k}=r_k$, 
 it can be deduced from \eqref{eq:fconda} and \eqref{eq:fcondma} that 
$B_1=-r_1 P_1$, 
where $P_1$ is a $I_1\times I_2$ matrix being the transition kernel of the update $i_2|i_1$ of the auxiliary process.
Similarly, one can show $B_2=-r_2P_2$, where $P_2$ is a $I_2\times I_1$ matrix being the transition kernel of the update $i_1|i_2$ of the auxiliary process. Hence, the autoregressive matrix of $\delta\fac[1](t)$ is  
$r_1r_2 P_1P_2$, where $P_1P_2$ is the transition kernel of the auxiliary Gibbs sampler $\{i(t)\}$. Consequently, the spectrum of the autoregressive matrix is $r_1 r_2 \lambda_i$, where $\lambda_i$ are the eigenvalues of $P_1 P_2$. The largest $|\lambda_i|$ is of course 1 since $P_1P_2$ is a stochastic matrix. However, 
since $\delta \fac[1]$ is constrained to have zero sum, by Lemma \ref{lemma:linear_constrain} in the Appendix the rate of convergence of $\delta\fac[1](t)$  is not given by the largest modulus eigenvalue of the autoregressive matrix, but the largest modulus eigenvalue whose eigenvector has zero sum, i.e., we need to consider only the subspace orthogonal to the vector of 1's. Therefore,  the rate of convergence of $\{\delta\facall(t)\}$ equals  $r_1r_2$ times the second largest modulus eigenvalue of $P_1P_2$, which is $\rho_{aux}$ by definition.
\end{proof}

With Proposition \ref{prop:bl_two_fac} in place, the main result of this Section on the mixing time of the algorithms follows immediately. 

\begin{theorem}\label{thm:bl_two_fac}
For balanced levels designs with $K=2$, the rate of convergence of the Gibbs Sampler and the collapsed Gibbs Sampler are given by, respectively,
\begin{align*}
\max\left\{
\frac{N\tau_{0}}{N\tau_{0}+I_1\tau_{1}}
,
\frac{N\tau_{0}}{N\tau_{0}+I_2\tau_{2}}
\right\}
\,,\qquad
\frac{N\tau_{0}}{N\tau_{0}+I_1\tau_{1}}
\frac{N\tau_{0}}{N\tau_{0}+I_2\tau_{2}}
\rho_{aux}\,,
\end{align*}
where $\rho_{aux}$ is the rate of convergence of the auxiliary Gibbs sampler $\{i(t)\}$ with invariant distribution $p(i_1,i_2)= n_{i_1i_2}/N$.
\end{theorem}
Note that if the design is in fact balanced cells, the rates given in Theorem \ref{thm:bl_two_fac} match those of Theorem \ref{thm:rate_gibbs_bl}, as they should, since $\rho_{aux}=0$ in this case. 

A corollary to this Theorem is that the mixing time of the Gibbs sampler is $1+\max\{\frac{N\tau_0}{I_1\tau_1},\frac{N\tau_0}{I_2\tau_2}\}$ and that of the collapsed Gibbs sampler is no larger than $1+\min\{\frac{N\tau_0}{I_1\tau_1},\frac{N\tau_0}{I_2\tau_2},T_{aux}\}$, where $T_{aux}$ is the mixing time of the auxiliary process $\{i(t)\}$. An implication of this is that  the collapsed Gibbs Sampler is never slower than the standard Gibbs Sampler and it is has good mixing both when the amount of data per level is low and high. To see this, 
note first the ratios $N/I_1$ and $N/I_2$ coincide with the number of datapoints per column and row, respectively, in the data incidence matrix with entries $n_{i_1i_2}$ and thus their value increases as the amount of data per level increases.
On the contrary the mixing time $T_{aux}$ of the auxiliary process $\{i(t)\}$ tends to decrease as the amount of data per level increases because the latter corresponds to adding more edges in the conditional independence graph, hence larger connectivity in the state space of the auxiliary process. Unfortunately, it is not true in general that the minimum across $\frac{N\tau_0}{I_1\tau_1}$, $\frac{N\tau_0}{I_2\tau_2}$ and $T_{aux}$ is uniformly bounded over $N$. 
Consider for example a design where users and items are split into two communities of equal size, and users inside each community have rated all items from their community and no item from the other community. 
In this case the random walk $\{i(t)\}$ is reducible.
Therefore $T_{aux}=\infty$ and, provided both $N/I_1$ and $N/I_2$ go to infinity, the mixing time of the collapsed Gibbs Sampler diverges as $N$ goes to infinity.

We now address the case of number of factors $K>2$ that Theorem \ref{thm:bl_two_fac} does not cover. 
A conjecture we make in this paper is that $1+\max_{k=1,\dots,K}\frac{N\tau_0}{I_k\tau_k}$ is the mixing time of the Gibbs sampler also for $K>2$.  We have experimented numerically quite extensively, since for specific examples we can compute the mixing time by computing numerically the largest eigenvalue of an explicit matrix, and we have not been able to find a counter-example. The missing step for a generic result would be to show that $\{\delta \facall(t)\}$ always mixes faster than $\{\mav(t)\}$. Such  a result would also immediately prove, due to Proposition \ref{prop:averages_rate}, that the collapsed Gibbs sampler has lower mixing time than the Gibbs sampler for arbitrary number of factors for balanced levels designs. On the other hand, numerical experimentation has also showed that certain extensions of Theorem \ref{thm:bl_two_fac} are not true. We know that the convergence rate of the collapsed Gibbs sampler can be larger than $\prod_{k=1,\dots,K}\frac{N\tau_0}{I_k\tau_k}$ when $K>2$; we also know that the rate will depend on the order that the different components are updated. We return to these points in the Discussion.

We close the section with some asymptotic considerations on the complexity. The following arguments assume that the mixing time of the Gibbs sampler is the conjectured  $1+\max_{k=1,\dots,K}\frac{N\tau_0}{I_k\tau_k}$; we will not consider the collapsed Gibbs sampler in the following considerations since we do not have conjecture for its rate when $K>2$. The asymptotic behaviour of the Gibbs sampler mixing time  depends on the regime under consideration as it was for balanced cells designs. The mixing time can be as bad as $\mathcal{O}(N)$, for example if the number of levels of at least one factor is fixed as $N$ grows; it can be $\mathcal{O}(N^{1-1/K})$ in the regime where $I_1=\ldots=I_K$ and $N=\mathcal{O}(I_1^K)$; but it can also be $\mathcal{O}(1)$ in the sparse observation regime where $N=I_1=\ldots=I_2$. The Appendix discuss the computational cost per iteration, which for these designs can grow quadratically with the number of parameters, as opposed to linearly in the case of balanced cells. In terms of its growth with the observations, this can be $\mathcal{O}(1)$, in infill asymptotics regimes where the number of levels of factors does not grow with $N$;
 it can be $\mathcal{O}(N^{2/K})$ when $I_1=\ldots=I_K$ and $N=\mathcal{O}(I_1^K)$; but 
 it can also be $\mathcal{O}(N)$ in the sparse regime $N=I_1=\ldots=I_2$. Connecting now to the observation in \cite{GaoOwen2017EJS}, we obtain that for $K=2$ when $N=\mathcal{O}(I_1^2)$ and $I_1=I_2$, the complexity of the Gibbs sampler is $\mathcal{O}(N^{3/2})$, hence the algorithm is not scalable.

\section{Simulation Studies}

\subsection{Simulated data with missingness completely at random}
\label{sec:simstudy}
First we consider simulated data with $K=2$ and $I_1=I_2$.
We assume data to be missing completely at random, where for each combination of factors we observe a datapoint, i.e., $n_{i_1i_2}=1$, with probability 0.1 independently of the rest, and otherwise we have a missing observation, i.e.,\ $n_{i_1i_2}=0$.
Since the mixing time of the samplers under consideration does not depend on the the value of the observations $y$, but only on their presence or absence, we can set $y_{i_1i_2}=0$ without affecting the computed convergence rates.
In this context our theory does not apply directly because the designs under consideration are not balanced in general.
However, we can still compute numerically the convergence rate of the Gibbs Sampler and its collapsed version in the context of known precisions, using the results discussed in Section \ref{sec:spectral-Gauss}, to explore to which extent the qualitative findings of our theory still apply.
Figure \ref{fig:two_fac_scaling} displays the behaviour of the mixing time of the Gibbs Sampler and its collapsed version in an outfill asymptotic regime, where both the number of datapoints and factor levels increase.
For the simulations we fixed the precision terms $\tau_k$ to 1 and take $I_1$ in the set $\{50,500,1000,2000\}$.
The results suggest that the mixing time of the Gibbs Sampler diverges with $N$, while the mixing time of its collapsed version converges to 1 as $N$ increases.
This is coherent with the theoretical results of previous section.
In fact, we can compare the mixing times that we computed numerically with the theoretical values computed as if the design were balanced levels, which of course it is not here. The figure shows an extremely close match, which showcases the use of our theory beyond the specific designs that have facilitated the analysis. 
This suggests that the theory previously developed is relevant beyond cases that strictly satisfy balanced levels.
Since the cost per iteration of both samplers is $\Op(N)$, the results in Figure \ref{fig:two_fac_scaling} suggest that, for the asymptotic regime considered in this section, the computational complexity of the Gibbs Sampler is $\Op(N^{3/2})$ and the one of the collapsed Gibbs Sampler is $\Op(N)$.

\subsection{ETH Instructor Evaluations dataset}\label{sec:InstEval}
We now consider a real dataset containing university lecture evaluations by students at ETH Zurich.
The dataset is freely available from the R package \emph{lme4} \citep{bates2015lme4} under the name \emph{InstEval}.
It contains $73421$ observations, each corresponding to a score ranging from 1 to 5, assigned to a lecture together with 6 factors potentially impacting such score, such as identity of the student giving the rating or department that offers the course.
See the \emph{lme4} help material for more details on the dataset.
We fit model \eqref{eq:anova} to the \emph{InstEval} dataset.
Following the notation in \eqref{eq:anova}, we have $N=73421$, $K=6$ and 
 $(I_1,\dots,I_K)=(2972,1128,4,6,2,14)$. Clearly, a categorical response calls for a generalised linear model extension of \eqref{eq:anova}, however the point of this analysis is to test the algorithms, and \eqref{eq:anova} is not an outright unreasonable model to fit for this dataset.

First we consider the known precision case, where the values $\tau_k$ are assumed to be known.
In this context our theory does not apply directly because the design of the dataset is not balanced.
However, we can still compute numerically the convergence rate of the Gibbs Sampler and its collapsed version, using the results discussed in Section \ref{sec:spectral-Gauss}.
Consider first a two-factor case, by restricting our attention to the first two factors.
In this case, setting $\tau_0=\tau_1=\tau_2=1$, the mixing times of the Gibbs Sampler and its collapsed version are, respectively, $68.9$ and $7.8$.
Such values are numerical approximations obtained by computing the autoregressive matrix $B$ explicitly and then using the power method to approximate the size of its largest eigenvalue. 
If instead we consider the first and the last factor, thus having $K=2$ and $(I_1,I_2)=(2972,14)$, the mixing times of the Gibbs Sampler and its collapsed version become, respectively, $5245.6$ and $4.8$.
This is coherent with our theory, which suggests that the presence of a factor with a small number of levels should severely slow down the Gibbs Sampler while not affecting the collapsed version.
We can compute the mixing time implied by Theorem \ref{thm:bl_two_fac}, even if the design is not balanced levels; the numbers we obtain for the Gibbs sampler are $66.1$ and $5245.4$ and the one for the collapsed Gibbs Sampler are $8.3$ and $5.0$, depending on whether $(I_1,I_2)=(2972,1128)$ or $(I_1,I_2)=(2972,14)$, respectively.
All values match closely the values obtained numerically. 
This  suggests that the theory previously developed can be highly informative also for unbalanced cases, provided the level of unbalancedness in the design is moderate.
Finally, if we fit the whole dataset, with $K=6$ and $\tau_k=1$ for $k=0,\dots,6$, the mixing times of the Gibbs Sampler and its collapsed version are, respectively, $36687.0$ and $137.2$. The mixing time of the Gibbs sampler implied by our theory, which again does not apply in this design, is  $36711.5$, which is accurate again. 
In this case the mixing time of the collapsed Gibbs Sampler, despite being orders of magnitude smaller than the non-collapsed version, is moderately large, suggesting that the residual chain $\{\delta \facall\}$ mixes slower than in the other examples.

Next consider the case of unknown precisions, where the hyperparameters $\tau_k$ are given a prior distribution and the posterior of interest is the joint distribution of $\facall$ and $\tau=(\tau_1,\dots,\tau_K)$.
We consider five Markov chain Monte Carlo schemes.
The first two schemes alternate sampling $\tau$ from the conditional distribution $\claw{\tau}{\facall}$ and updating $\facall$ with the Gibbs Sampler and its collapsed version, respectively.
These are the most straightforward extensions of the samplers studied above to the unknown precisions case.
Provided conjugate priors are used, the update $\tau\sim \claw{\tau}{\facall}$ is trivial as the precision terms $\tau_k$ are conditionally independent given $\facall$.
The third and fourth schemes combine the first and second schemes, respectively, with the parameter expanded data augmentation methodology \citep{liu1999parameter,meng1999seeking}. 
In the context under consideration, the parameter expanded methodology seeks to avoid issues related to potential correlation between the two blocks $\facall$ and $\tau$ by introducing appropriate auxiliary parameters, see \cite{gelman2008using} for more discussion.
Finally, the fifth scheme is the No U-Turn sampler \citep{hoffman2014no}, a state-of-the-art Hamiltonian Monte Carlo scheme implemented in the \emph{R} package \emph{RStan} \citep{RStan}.
For the precision parameters, we used a standard flat prior $p(\tau_k^{-1/2})\propto 1$, mainly to facilitate the implementation of parameter expanded methodologies.
In order to avoid potential issues related to using flat priors with a very low number of factor levels, we excluded the factor with only two levels from the analysis, resulting in $K=5$ and $(I_1,\dots,I_K)=(2972,1128,4,6,14)$.
\begin{table}[h!]
\centering
\begin{tabular}{c|c|cccccccccc}
  \hline
\multirow{ 2}{*}{Scheme} & time per &  \multicolumn{2}{c}{Effective Sample Size / time (1/s)}\\ 
 & 1000 iter. & $(\fac[0],\,\ma[1],\,\ma[2],\,\ma[3],\,\ma[4],\,\ma[5])$& 
 $(\sigma_0,\sigma_1,\sigma_2,\sigma_3,\sigma_4,\sigma_5)$\\ 
  \hline
GS & 13.2s & 
(0.07,\,11.0,\,2.12,\,0.16,\,0.21,\,0.87) & 
(60.9,\,15.9,\,36.8,\,3.56,\,2.53,\,2.14)\\
cGS & 14.2s &
(65.9,\,42.1,\,18.1,\,70.5,\,62.3,\,35.0) &
(55.1,\,14.7,\,34.5,\,17.7,\,33.9,\,2.51) \\
GS+PX & 13.5s & 
(0.06,\,10.7,\,2.02,\,0.08,\,0.11,\,0.95) &
(59.6,\,41.2,\,43.9,\,0.85,\,0.58,\,2.33)\\
cGS+PX & 14.4s & 
(62.5,\,44.1,\,19.9,\,62.9,\,63.2,\,34.6) &
(55.1,\,38.0,\,41.9,\,19.2,\,33.0,\,2.96)\\
HMC & 1112.6s & 
(0.11,\,0.78,\,0.19,\,0.08,\,0.25,\,0.68) &
(0.99,\,0.51,\,0.99,\,0.10,\,0.41,\,0.18)\\ 
\hline
\end{tabular}
\caption{
Comparison of sampling schemes on the \emph{InstEval} data, where 
$\sigma_k=1/\sqrt{\tau_k}$. GS and cGS refer to the Gibbs Sampler and the collapsed version with precision updates, while +PX indicates combination with the parameter expanded methodology. HMC referes to the \emph{RStan} implementation.
Numbers are averaged over 10 runs of 10000 iterations for each scheme, discarding the first 1000 samples as burn-in.
}\label{table:ESS}
\end{table}
Table \ref{table:ESS}  and Figure \ref{fig:ACF} report runtimes for the five schemes together with effective sample sizes and autocorrelation functions.
It can be seen that the first four schemes have similar runtimes, but the ones using the collapsed methodology proposed in this paper induce a much faster mixing compared to the others. 
The use of the parameter expansion methodology provides a further, very limited in this case, improvement.
\begin{figure}[h!]
\includegraphics[width=\linewidth]{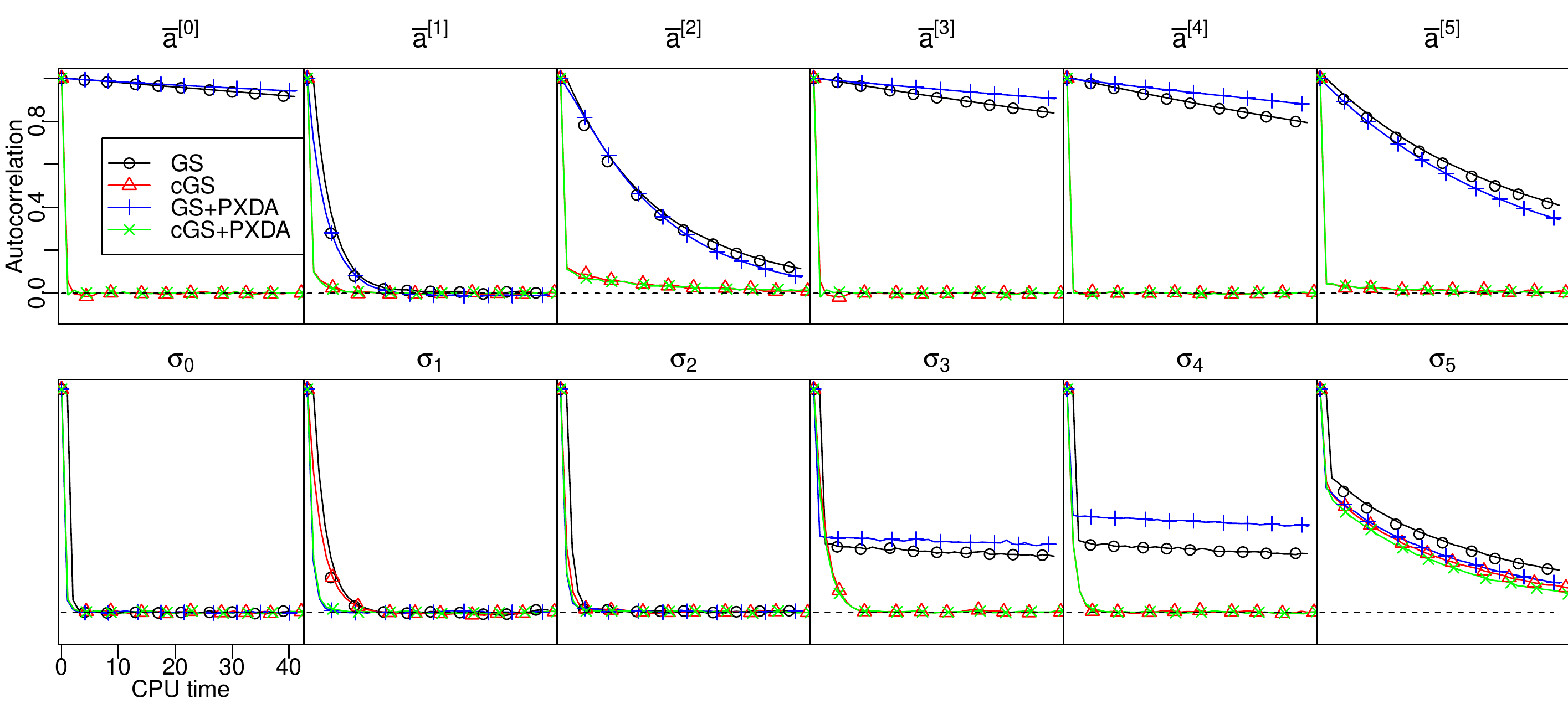}
\caption{Autocorrelation functions of $\ma[k]$ and $\sigma_k=\tau_k^{-1/2}$, for $k=1,\dots,K$, on the \emph{InstEval} dataset.
Lines are averaged over 10 runs of 10000 iterations for each scheme.}\label{fig:ACF}
\end{figure}
On the other hand, Hamiltonian Monte Carlo has a cost per iteration that is two orders of magnitude larger than the other schemes, resulting in the lowest effective sample sizes per unit of computation time.
Finally, to obtain a higher level sense of the practicality of the approach we pursue in this article, we also fit the same crossed effect model in a frequentist fashion using the \emph{R} package \emph{lme4}, which took $40.9$ seconds to run.
All computations were performed on the same desktop computer with 16GB of RAM and an $i7$ Intel processor.
It is worth noting that the first four schemes were directly implemented using a high level language such as $R$, so we would expect significant further speed-ups by using a low-level language and use of distributed computing for the precomputations needed for the Gibbs samplers.

\section{Discussion}

There are many directions this work can move forward. We highlight the two that are most imminent. First is to investigate the conjecture made in Section \ref{sec:balanced-levels} that the mixing time of the Gibbs sampler for balanced levels designs is $1+\max_{k=1,\dots,K}\frac{N\tau_0}{I_k\tau_k}$. If this is true we also obtain that the collapsed Gibbs sampler has always smaller rate for such designs. The second is to obtain a characterisation of the rate of the collapsed Gibbs sampler for such designs when $K>2$. From numerical experimentation we know that the natural extension of the expression of Theorem \ref{prop:bl_two_fac} is not true for $K>2$, hence a different line of attack is needed. 

\section*{Acknowledgement}
The authors would like to acknowledge helpful discussions with Art B.\ Owen.
Papaspiliopoulos acknowledges financial support from the
Spanish Ministry of Economics via a research grant.
Zanella was supported by the European Research Council (ERC) throught the 
``New Directions in Bayesian NonParameterics'' StG 306406.

\section*{Appendix}

\subsection*{Proof of Theorems and auxiliary results}
\begin{proof}[of Theorem \ref{th:multigrid1}]
  For concreteness and without affecting the validity of the argument we assume that the algorithm updates factors and their levels in ascending order, i.e., first simulates $\fac[0]$, then $\fac[1]_1$, $\fac[1]_2$, and so on and so forth.  We first establish the result for the Gibbs sampler,i.e., part 1.  Note that due to the conditional independence structure the algorithm can be equivalently represented as one that samples in blocks according to the conditional laws $\claw{\fac}{y,\fac[-k]}$. For each iteration $t$, each such draw, $\fac(t)$ can be transformed to $\ma(t)$ and $\delta \fac(t)$. Proposition \ref{prop:b_l} establishes that the
  \begin{align}
    & \claw{\ma(t),\delta \fac(t)}{y,\fac[0](t),\ldots,\fac[k-1](t),\fac[k+1](t-1),\ldots,\fac[K](t-1)} = \nonumber\\
    & \claw{\ma(t)}{y,\ma[0](t),\ldots,\ma[k-1](t),\ma[k+1](t-1),\ldots,\ma[K](t-1)} \times \label{eq:mav_conditionals_gibbs}\\
    & \claw{\delta \fac(t)}{y,\delta \fac[1](t),\ldots,\delta \fac[k-1](t),\delta \fac[k+1](t-1),\ldots,\delta \fac[K](t-1)}\,.\nonumber
  \end{align}
  Appealing to the equivalence of local and global Markov properties, as in Section 3 of \cite{besag}, we obtain that the processes $\{\mav(t)\}$ and $\{\delta \facall(t)\}$, obtained as functions of $\{\facall(t)\}$, are each a Markov chain with respect to its own filtration, and independent of each other. 

The collapsed Gibbs Sampler case is analogous.
Here the sampler iterates the updates of $\claw{\fac}{y,\fac[-0,-k]}$ for $k=1,\dots,K$.
It can be easily deduced from Proposition 1 that  $\claw{\mav^{(-0)},\delta \facall}{y}=\claw{\mav^{(-0)}}{y} \claw{\delta \facall}{y}$.
Therefore, transforming each draw $\fac[k](t)$ to $\ma(t)$ and $\delta \fac(t)$, we obtain
  \begin{align*}
    & \claw{\ma(t),\delta \fac(t)}{y,\fac[1](t),\ldots,\fac[k-1](t),\fac[k+1](t-1),\ldots,\fac[K](t-1)} = \\
    & \claw{\ma(t)}{y,\ma[1](t),\ldots,\ma[k-1](t),\ma[k+1](t-1),\ldots,\ma[K](t-1)} \times\\
    & \claw{\delta \fac(t)}{y,\delta \fac[1](t),\ldots,\delta \fac[k-1](t),\delta \fac[k+1](t-1),\ldots,\delta \fac[K](t-1)}\,.
  \end{align*}
  It follows that the processes $\{\mav^{(-0)}(t)\}$ and $\{\delta \facall(t)\}$, obtained as functions of $\{\facall(t)\}$, are each a Markov chain with respect to its own filtration, and independent of each other. 
\end{proof}

\begin{proof}[of Lemma \ref{lemma:transformation_MC}]
The Markovianity of $\{y(t)\}$ follows from the fact that the $\sigma$-algebras associated to $x(t)$ and $y(t)$ coincide.
Denote by $\mathcal{X}$ the state space $\{x(t)\}$ and by $P$ its transition kernel.
Similarly $\mathcal{Y}$ and $Q$ for $\{y(t)\}$.
By taking $\mathcal{Y}=\phi(\mathcal{X})$ we can assume $\phi$ to be invertible without loss of generality.
For every $t\geq 1$, $P^t$ and $Q^t$ are integral operators on $L^2(\pi)$ and $L^2(\mu)$, where $\mu$ is the pushforward of $\pi$ under $\phi$ defined as $\mu(A)=\pi(\phi^{-1}(A))$ for every measurable $A$.
From $y(t)=\phi(x(t))$ it follows $Q^t(\psi(f))=\psi(P^t(f))$, where $\psi:f\mapsto f\circ \phi^{-1}$ is the linear map defined by $f\circ \phi^{-1}(y)=f(\phi^{-1}(y))$.
The equality between the rates of convergence of $P$ and $Q$ follows by noting that $\psi$ is an isomorphism from $L^2(\pi)$ to $L^2(\mu)$.
\end{proof}


\begin{proof}[of Lemma \ref{lemma:independent_MC}]
By the independence assumption, the integral operator $P$ associated to $\{x(t)\}$ equals the tensor product $P_1\otimes P_2$, where  $P_1$ and $P_2$ are the integral operators on $L^2(\pi_1)$ and $L^2(\pi_2)$ associated to $\{x_1(t)\}$ and $\{x_2(t)\}$.
It follows that the spectrum $\sigma(P)$ is the product $\sigma(P_1)\sigma(P_2)$ of the two spectra of $P_1$ and $P_2$ (see e.g. \cite{brown1966spectra}).
Since the largest modulus eigenvalue in both $\sigma(P_1)$ and $\sigma(P_2)$ is 1 it follows that the second largest modulus eigenvalue in $\sigma(P_1)\sigma(P_2)$ equals the maximum of the second largest modulus eigenvalues in $\sigma(P_1)$ and $\sigma(P_2)$.
\end{proof}

\begin{lemma}\label{lemma:linear_constrain}
Let $\{x(t)\}$ be a $d$-dimensional gaussian AR(1) process with $\mathbb{E}[x(t+1)\mid x(t)]=Bx(t)+b$, for some fixed $b$, and stationary distribution $N(\mu,\Sigma)$ concentrated on the hyperplane $\sum_{i} x_i=0$ and $\Sigma$ of rank $d-1$.
Then the rate of convergence of $\{x(t)\}$ equals the largest modulus eigenvalue of $B$ whose eigenvector has zero sum.
\end{lemma}
\begin{proof}
The proof boils down to considering $\{x_{-d}(t)\}$ and applying the classical, non-singular version of this theorem, see \cite[Thm.1]{roberts1997updating}.
More precisely, Lemma \ref{lemma:transformation_MC} and the constrain $\sum_i x_i=0$ imply that $\{x_{-d}(t)\}$ is a Markov chain with the same rate of convergence as $\{x(t)\}$.
Then, from 
\begin{align*}
\mathbb{E}[x_i(t+1)\mid x_{-d}(t)]
&=
\mathbb{E}[x_i(t+1)\mid x_{-d}(t),x_{d}(t)=-\sum_{j\neq d}x_{j}(t)]
\\&=
\sum_{j\neq d}B_{ij}x_{j}(t)+B_{id}(-\sum_{j\neq d}x_{j}(t))+b_i
=
\sum_{j\neq d}(B_{ij}-B_{id})x_{j}(t)+b_i\,,
\end{align*}
it follows that $\mathbb{E}[x_{-d}(t+1)\mid x_{-d}(t)]=\tilde{B}x_{-d}(t)+b_{-d}$ with $\tilde{B}_{ij}=B_{ij}-B_{id}$ for all $i,j\in\{1,\dots,d-1\}$.
Thus, by \citet[Thm.1]{roberts1997updating}, the rate of convergence of $\{x_{-d}(t)\}$ equals the largest modulus eigenvalue of $\tilde{B}$.
To conclude we show that if $v$ is an eigenvector of $B$ such that $\sum_iv_i=0$ it follows that $v_{-d}$ is an eigenvector of $\tilde{B}$ with the same eigenvalue.
Indeed, if $B v=\lambda v$ and $\sum_iv_i=0$ it follows
\[
(\tilde{B}v_{-d})_i
=\sum_{j\neq d}\tilde{B}_{ij}v_j
=\sum_{j\neq d}B_{ij}v_j-B_{id}\sum_{j\neq d}v_j
=\sum_{j}B_{ij}v_j
=\lambda v_i\,.
\]
\end{proof}

\subsection*{Cost per iteration of the Gibbs Sampler and its collapsed version}
In order to implement the Gibbs Sampler, the computation of the one and two-dimensional marginals $\{\datc[k]\}$ and $\{\datc[l,k]\}$ of the data incidence table are required, as well as the computation of the weighted averages $\{\wmy^{(k)}_j\}$ of the data.
Such precomputation needs to be performed only once and requires $\mathcal{O}(N)$ 
 operations in general.
Then, at each iteration of the Gibbs Sampler the update of $\fac[0]$ and each $\fac[k]_j$ can be accomplished in $\mathcal{O}(\sum_{l}I_l)$  and $\mathcal{O}(\sum_{l\neq k}I_l)$ operations using \eqref{eq:fcondmu} and \eqref{eq:fconda}, respectively, resulting in a total of $\mathcal{O}(\sum_{k}I_k\sum_{l\neq k}I_l)$ operations for each Gibbs sweep.
The latter can be as bad as $\mathcal{O}(p^2)$, where $p=\sum_k I_k$ is the number of parameters and its relationship with $N$ depends on the asymptotic regime under consideration.

For the collapsed Gibbs Sampler one needs to additionaly precompute $\{s^{(k)}\}$ defined in Proposition \ref{prop:marginal_law}, which can be done in $\mathcal{O}(p)$ operations given $\{\datc[k]\}$. Therefore the collapsed Gibbs Sampler has a precomputation cost of order $\mathcal{O}(N)$, 
 similarly to the standard Gibbs Sampler.
Moreover, the updates of $\fac[0]$ from \eqref{eq:fcondmu_collapsed} for $k=1,\dots,K$ require $\mathcal{O}(\sum_{k}I_k\sum_{l\neq k}I_l)$ operations altogether, which is at most $\mathcal{O}(p^2)$.
Thus the collapsed Gibbs Sampler has also the same cost per iteration of the standard Gibbs Sampler.

In the balanced cells case, the only precomputation required is the one of $\{\wmy^{(k)}_j\}$, which has $\mathcal{O}(N)$ cost.
Also, each Gibbs or collapsed Gibbs sweep can be accomplished in $\mathcal{O}(p)$ operations, rather than $\mathcal{O}(p^2)$, using \eqref{eq:fcondmu-bl}, \eqref{eq:fconda-cells} and the version of \eqref{eq:fcondmu_collapsed} for balanced cells.


\bibliographystyle{biometrika}
\bibliography{crossed}

\end{document}